\documentclass[11pt]{article}
\usepackage{times}
\usepackage{amsmath,amsthm,amsfonts}
\usepackage{fullpage}

\begin{document}

\title{Better short-seed quantum-proof extractors}

\author{
Avraham Ben-Aroya\thanks{The Blavatnik School of Computer Science,
Tel-Aviv University, Tel-Aviv 69978, Israel. Supported by the Adams
Fellowship Program of the Israel Academy of Sciences and Humanities,
by the Israel Science Foundation, by the Wolfson Family Charitable
Trust, and by a European Research Council (ERC) Starting Grant. }
\and Amnon Ta-Shma\thanks{The Blavatnik School of Computer Science,
Tel-Aviv University, Tel-Aviv 69978, Israel. Supported by the FP7 FET-Open project QCS} }
\date{}

\maketitle

\newcommand{\remove}[1]{}

\newcommand{\set}[1]{{\left\{ #1\right\}}}
\newcommand{\B}[0]{{\left\{0,1\right\}}}
\newcommand{\zo}{\set{0,1}}

\newcommand{\abs}[1]{\left| #1 \right|}
\newcommand{\norm}[1]{\left\| #1 \right\|}
\newcommand{\half}{\frac 1 2}
\newcommand{\minent}{{H_{\infty}}}
\newcommand{\Tr}{\mathrm{Tr}}
\newcommand{\eps}{{\epsilon}}
\newcommand{\logeps}{\log \epsilon^{-1}}
\newcommand{\eqdef}{\stackrel{\rm def}{=}}
\newcommand{\wh}[1]{\widehat{#1}}
\newcommand{\wt}[1]{\widetilde{#1}}
\newcommand{\poly}{{\rm poly}}
\newcommand{\pw}[1]{{\lfloor\!\!\lfloor#1\rfloor\!\!\rfloor}}
\newcommand{\brho}{{\bar{\rho}}}
\newcommand{\ol}[1]{\overline{#1}}

\newcommand{\ket}[1]{\left|#1\right\rangle}
\newcommand{\bra}[1]{\left\langle #1\right|}
\newcommand{\braket}[2]{\left.\left\langle #1\right|#2\right\rangle}
\newcommand{\ketbra}[2]{\ket{#1}\!\bra{#2}}
\newcommand{\tensor}{{\otimes}}
\newcommand{\tnorm}[1]{\norm{#1}_{\mathrm{tr}}}

\newcommand{\enc}[3]{{#1} \stackrel{{#3}}{\mapsto} {#2}}
\newcommand{\bn}{{\bar{n}}}
\newcommand{\Hom}{\mathrm{Hom}}

\newcommand{\cC}{\mathcal{C}}
\newcommand{\cD}{\mathcal{D}}
\newcommand{\cE}{\mathcal{E}}
\newcommand{\rE}{\mathrm{E}}
\newcommand{\cF}{\mathcal{F}}
\newcommand{\cH}{\mathcal{H}}
\newcommand{\cQ}{\mathcal{Q}}

\newcommand{\NW}{\mathrm{NW}}
\newcommand{\TR}{\mathrm{TR}}
\newcommand{\E}{\mathop{\mathbb{E}}\displaylimits}
\newcommand{\sam}{{\sim}}

\renewcommand{\qed}{\hfill{\rule{2mm}{2mm}}}
\renewenvironment{proof}[1][]{\begin{trivlist}
\item[\hspace{\labelsep}{\bf\noindent Proof#1:\/}] }{\qed\end{trivlist}}



\newenvironment{Proof}{{\bf Proof:\ }}{\hfill$\Box$\medskip}

\newtheorem{theorem}{Theorem}[section]
\newtheorem{definition}{Definition}[section]
\newtheorem{lemma}{Lemma}[section]
\newtheorem{claim}{Claim}[section]
\newtheorem{proposition}{Proposition}[section]
\newtheorem{corollary}{Corollary}[section]
\newtheorem{example}{Example}[section]
\newtheorem{fact}{Fact}[section]
\newtheorem{remark}{Remark}[section]

\begin{abstract}
We construct a strong extractor against quantum storage that works
for every min-entropy $k$, has logarithmic seed length, and outputs
$\Omega(k)$ bits, provided that the quantum adversary has at most
$\beta k$ qubits of memory, for any $\beta < \half$.  The construction works by first
condensing the source (with minimal entropy-loss) and then applying
an extractor that works well against quantum adversaries when the source is close to uniform.

We also obtain an improved construction of a strong quantum-proof
extractor in the high min-entropy regime. Specifically, we construct
an extractor that uses a logarithmic seed length and extracts
$\Omega(n)$ bits from any source over $\B^n$, provided that the
min-entropy of the source conditioned on the quantum adversary's
state is at least $(1-\beta) n$, for any $\beta < \half$.
\end{abstract}

\setcounter{page}{1}
\pagestyle{plain}

\section{Introduction}
\label{sec:introduction}

In the  \emph{privacy amplification} problem Alice and Bob share
information that is only partially secret with respect to an
eavesdropper Charlie. Their goal is to distill this information to a
shorter string that is completely secret. The problem was introduced
in~\cite{BBR88,BBCM95} for classical eavesdroppers. An interesting
variant of the problem, where the eavesdropper is allowed to keep
quantum information rather than just classical information, was
introduced by K\"{o}nig, Maurer and Renner~\cite{KMR05}. This
situation naturally occurs in analyzing the security of some quantum
key-distribution protocols~\cite{CRE04} and in bounded-storage
cryptography~\cite{KT08, KR07}.

The shared information between Alice and Bob is modeled as a
shared string $x \in \B^n$, sampled according a distribution $X$.
The information of the eavesdropper is modeled as a mixed state, $\rho(x)$,
which might correlated with $x$.

The privacy amplification problem can be solved by Alice and Bob,
but only by using a (hopefully short) random seed $y$, which can be
public. Thus, Alice and Bob look for a function  $E:\B^n \times \B^t
\to \B^m$ that acts on their shared input $x$ and the public random
string $y$, and extracts ``true randomness" for any ``allowed"
classical distribution $X$ and side information $\rho(X)$. More
formally, $E$ is an $\eps$-\emph{strong extractor for a family of
inputs $\Omega$}, if for any distribution $X$ and any quantum system
$\rho$ such that $(X;\rho) \in \Omega$, the distribution $Y \circ
E(X,Y) \circ \rho$ is $\eps$-close to $U \circ \rho$, where $U$
denotes the uniform distribution. (See Section~\ref{sec:def:ext} for
precise details.)

Clearly, no randomness can be extracted if, for every $x$, it is possible to recover $x$ from the side information
$\rho(x)$. We say the \emph{conditional
min-entropy} of $X$ with respect to $\rho(X)$ is $k$, if  an
adversary holding the state $\rho(x)$ cannot guess the string $x$
with probability higher than $2^{-k}$. Roughly
speaking, if one can extract $k$ almost uniform bits from a source
$X$ in spite of the side information $\rho(X)$, then the state $X
\circ \rho(X)$ is close to another state with conditional
min-entropy at least $k$.\footnote{Such a source is said to have
conditional \emph{smooth} min-entropy $k$.} Thus, in a very concrete
sense, the ultimate goal is finding extractors for sources with high
conditional min-entropy.\footnote{A simple argument shows an
extractor for sources with high conditional min-entropy is also an
extractor for sources with high conditional smooth min-entropy.} We
say $E$ is a \emph{quantum-proof} $(n,k,\eps)$ strong extractor if
it extracts randomness from every input $(X;\rho)$ with conditional
min-entropy at least $k$.

Not every classical extractor\footnote{We refer to extractors that extract randomness when the side information is classical as classical extractors.} is quantum-proof, as was shown by Gavinsky et
al.~\cite{GKKRW07}. On the positive side, several well-known
classical extractors are quantum-proof.
Table~\ref{table:classical-extractors} lists some of these
constructions. We remark that the best explicit classical extractors
\cite{GUV07,DW08,DKSS09} achieve significantly better parameters
than those known to be quantum-proof.

\begin{table}[t]
\begin{tabular}{|l|l|l|l|}\hline
no. of  truly&
no. of &
classical & quantum-proof \\
random bits &
output bits & & \\
\hline\hline
 $O(n)$ &
$m=k-O(1)$ & Pair-wise independence, \cite{ILL89} & \checkmark \cite{KMR05} \\
$O(n-k+\log n)$ &
$m=n$ & Fourier analysis, collision \cite{DS05} & \checkmark \cite{FS07} \\
 $\Theta(m)$ &
$m = k-O(1)$ & Almost pair-wise ind., \cite{SZ99,GW97} & \checkmark, \cite{TSSR10}\\
 $O({\log^2n \over \log (k)})$ & $k^{1-\zeta}$ &
Designs, \cite{T01} & \checkmark, \cite{DV10} \\
$O(\log n)$ & $m=\Omega(n)$ & \cite{NZ96,CRVW02} & \checkmark, This
paper, provided $k > (\half+\zeta) n$
\\
\hline
$\log n+O(1)$ &
$m=k-O(1)$ &
Lower bound  \cite{NZ96,RT00} & \checkmark \\
\hline\hline
\end{tabular}
\caption{Explicit quantum-proof $(n,k,\eps)$ strong extractors. To simplify parameters, the error $\epsilon$ is a constant. }
\label{table:classical-extractors}
\end{table}

A simpler adversarial model is the ``bounded storage model" where
the adversary may store a limited number of qubits. The only
advantage of the bounded storage model for extractors is that it
simplifies the proofs, and allows us to achieve results which
currently we cannot prove in the general model. We say $E$ is an
$(n,k,b,\eps)$ strong extractor \emph{against quantum storage} if it
extracts randomness from every pair $(X;\rho)$ for which $X$ has at
least $k$ min-entropy and for every $x$, $\rho(x)$ is a mixed state with at most $b$
qubits.

In this paper we work with a slight generalization of the bounded storage model. We say $E$ is a \emph{quantum-proof} $(n,f,k,\eps)$ strong extractor for \emph{flat distributions} if it extracts randomness from every input $(X;\rho)$ for which $X$ is a flat distribution (meaning it is uniform over its support) with exactly $f$ min-entropy and the conditional min-entropy is at least $k$. In Lemma \ref{lem:flat-to-storage} we prove the easy observation that any quantum-proof $(n,f,k,\eps)$ strong extractor for flat distributions
is also a $(n,f,f-k,\eps)$ strong extractor against quantum storage.

We show a generic reduction from the problem of constructing
quantum-proof $(n,f,k,\eps)$ strong extractors for flat
distributions to the problem of constructing quantum-proof
$((1+\alpha)f,f,k,\eps)$ strong extractors for flat distributions,
and a similar reduction for the bounded storage model. In other
words, in our model the quantum adversary may have two types of
information about the source: first, it may have some classical
knowledge about it, reflected in the fact that the input $x$ is
taken from some classical flat distribution $X$, and second, it
holds a quantum state that contains some information about the
source. The reduction shows that without loss of generality we may
assume the classical input distribution is almost uniform. The
reduction uses a purely classical object called a \emph{strong
lossless condenser} and extends work done in \cite{TUZ07} on
extractors to quantum-proof extractors. This reduction holds for any
setting of the parameters.

We then augment this with a simple construction that shows how to
obtain a quantum-proof $((1+\alpha)f,f,k=(1-\beta)f,\eps)$  strong extractor for flat distributions,
provided that $\beta < \half$. The
argument here builds on work done in \cite{NZ96} on composition of
extractors and extends it to quantum-proof extractors. Together, these two reductions give:

\begin{theorem}
\label{thm:main-flat}
For any $\beta < \half$ and $\eps \ge 2^{-k^\beta}$, there exists an explicit quantum-proof $(n,k,(1-\beta) k,\eps)$ strong extractor for flat sources $E:\B^n \times \B^t \to \B^m$ with seed length $t=O(\log n+\logeps)$ and output length $m=\Omega(k)$.
\end{theorem}

Consequently,

\begin{theorem}
\label{thm:main-storage}
For any $\beta < \half$ and $\eps \ge 2^{-k^\beta}$, there exists an explicit $(n,k,\beta k,\eps)$ strong extractor against quantum storage, $E:\B^n \times \B^t \to \B^m$, with seed length $t=O(\log n+\logeps)$ and output length $m=\Omega(k)$.
\end{theorem}

This gives the first logarithmic seed length extractor against $b$
quantum storage that works for every min-entropy $k$ and extracts a
constant fraction of the entropy, and it is applicable whenever
$b=\beta k$ for $\beta<\half$.

We would like to stress that in most practical applications, and in particular in cryptographic applications such as quantum key distribution, it is generally impossible to bound the \emph{size} of the side information. For example, in quantum key distribution where extractors are used for privacy amplification, the conditional min-entropy of the source can be estimated by measuring the noise on the channel, whereas any estimate on the adversary's memory is an unproven assumption. Thus, an extractor proven to work only against quantum storage cannot be used in quantum key distribution protocols. We nevertheless feel that proving a result in the bounded storage model may serve as a first step towards solving the general question.

In fact, the second component in the above construction also works in the general quantum-proof setting. Specifically, this gives an extractor with seed length $t=O(\log n+\logeps)$ that extracts $\Omega(n)$ bits from any source with conditional min-entropy at least $(1-\beta)n$ for $\beta<\half$.

\begin{theorem}
\label{thm:main-qproof}
For any $\beta < \half$ and $\eps \ge 2^{-n^\beta}$, there exists an explicit quantum-proof $(n,(1-\beta)n,\eps)$ strong extractor $E:\B^n \times \B^t \to \B^m$, with seed length $t=O(\log n+\logeps)$ and output length $m=\Omega(n)$.
\end{theorem}

The rest of the paper is organized as follows.
Section~\ref{sec:preliminaries} contains all the necessary
preliminaries, including the formal definitions of min-entropy,
quantum-proof extractors and extractors against quantum storage. In
Section~\ref{sec:reduction} we give the reduction which shows it is
sufficient to construct extractors for sources with nearly full
min-entropy, when working in the bounded storage or flat sources
settings. In Section~\ref{sec:ext-less-than-half} we describe the
construction of quantum-proof extractors when the conditional
min-entropy is more than half, and give the proof of
Theorem~\ref{thm:main-qproof}. The proofs of
Theorems~\ref{thm:main-flat} and~\ref{thm:main-storage} are given in
Section~\ref{sec:final}.

\section{Preliminaries}
\label{sec:preliminaries}

\paragraph{Distributions.}
A distribution $D$ on $\Lambda$ is a function $D:\Lambda \to
[0,1]$ such that $\sum_{a \in \Lambda} D(a)=1$. We denote by $x
\sam D$ sampling $x$ according to the distribution $D$. Let $U_t$
denote the uniform distribution over $\B^t$. We measure the distance
between two distributions with the variational distance $
|D_1-D_2|_1 = \half \sum_{a \in \Lambda} |D_1(a)-D_2(a)|.$ The
distributions $D_1$ and $D_2$ are \emph{$\eps$-close} if
$|D_1-D_2|_1 \le \eps$.


The min-entropy of $D$ is denoted by $\minent(D)$ and is defined
to be $$\minent(D) = \min_{a: D(a)>0} -\log(D(a)).$$ If
$\minent(D) \ge k$ then for all $a$ in the support of $D$ it holds
that $D(a) \le 2^{-k}$. A distribution is \emph{flat} if it is
uniformly distributed over its support. Every distribution $D$
with $\minent(D) \ge k$ can be expressed as a convex combination
$\sum \alpha_i D_i$ of flat distributions $\set{D_i}$, each with
min-entropy at least $k$. We sometimes abuse notation and identify
a set $X$ with the flat distribution that is uniform over $X$.

If $X$ is a distribution over $\Lambda_1$ and $f: \Lambda_1 \to
\Lambda_2$ then $f(X)$ denotes the distribution over $\Lambda_2$
obtained by sampling $x$ from $X$ and outputting $f(x)$. If $X_1$
and $X_2$ are \emph{correlated} distributions we denote their joint
distribution by $X_1 \circ X_2$. If $X_1$ and $X_2$ are
\emph{independent} distributions we replace $\circ$ by $\times$ and
write $X_1 \times X_2$.

\paragraph{Mixed states.}
A pure state is a vector in some Hilbert space. A general quantum
system is in a {\em mixed state\/} --- a probability distribution
over pure states. Let $\{p_i, \ket{\phi_i}\}$ denote the mixed
state where the pure state~$\ket{\phi_i}$ occurs with
probability~$p_i$. The behavior of the mixed state
$\set{p_i,\ket{\phi_i}}$ is completely characterized by its \emph{
density matrix\/}~$\rho = \sum_i p_i \ketbra{\phi_i}{\phi_i}$, in
the sense that two mixed states with the same density matrix have
the same behavior under any physical operation. Notice that a
density matrix over a Hilbert space $\cH$ belongs to
$\Hom(\cH,\cH)$, the set of linear transformation from $\cH$ to
$\cH$. Density matrices are positive semi-definite operators and
have trace $1$.

The \emph{trace distance} between density matrices $\rho_1$ and
$\rho_2$ is $\tnorm{\rho_1-\rho_2} = \half \sum_i |\lambda_i|$,
where $\set{\lambda_i}$ are the eigenvalues of $\rho_1-\rho_2$.
The trace distance coincides with the variational distance when
$\rho_1$ and $\rho_2$ are classical states ($\rho$ is classical if it is diagonal in the standard basis).
Similarly to
probability distributions, the density matrices $\rho_1$ and
$\rho_2$ are \emph{$\eps$-close} if the trace distance between
them is at most $\eps$.

A positive operator valued measure (POVM) is the most general
formulation of a measurement in quantum computation. A POVM on a
Hilbert space $\cH$ is a collection $\set{F_i}$ of positive
semi-definite operators $F_i:\Hom(\cH,\cH) \to \Hom(\cH,\cH)$ that
sum-up to the identity transformation, i.e., $F_i \succeq 0$ and
$\sum F_i=I$. Applying a POVM $F=\set{F_i}$ on a density matrix
$\rho$ results in the distribution $F(\rho)$ that outputs $i$ with
probability $\Tr(F_i \rho)$.

A Boolean measurement $\set{F,I-F}$ \emph{$\eps$-distinguishes}
$\rho_1$ and $\rho_2$ if $|\Tr(F \rho_1) - \Tr(F \rho_2)| \ge
\eps$.

We shall need the following facts regarding the trace distance.
\begin{fact}\label{fact:trace-norm-distinguisher}
If $\tnorm{\rho_1 -\rho_2} = \delta$ then there exists a Boolean
measurement that $\delta$-distinguishes $\rho_1$ and $\rho_2$.
\end{fact}
\begin{fact}\label{fact:trace-norm-reduced-by-processing}
If $\rho_1$ and $\rho_2$ are $\eps$-close then $\cE(\rho_1)$ and $\cE(\rho_2)$ are $\eps$-close, for any physically realizable transformation $\cE$.
\end{fact}

\subsection{Min-entropy}

To define the notion of quantum-proof extractors we first need the notion
of quantum encoding of classical states.
\begin{definition}
Let $X$ be a distribution over some set $\Lambda$.
\begin{itemize}
  \item  An \emph{encoding} of $X$ is a collection $\rho = \set{\rho(x)}_{x
\in \Lambda}$ of density matrices.
  \item An encoding $\rho$ is a \emph{$b$-storage encoding} if $\rho(x)$ is a mixed state over $b$ qubits, for all $x \in \Lambda$.
  \item An encoding is \emph{classical} if $\rho(x)$ is classical
for all $x$.
\end{itemize}
\end{definition}
The average encoding is denoted by $\brho_X = \E_{x \sam X} [\rho(x)]$.

Next we define the notion of conditional min-entropy. The
conditional min-entropy of $X$ given $\rho(X)$ measures the average
success probability of predicting $x$ given the encoding $\rho(x)$.
Formally,
\begin{definition}
The \emph{conditional min-entropy of $X$ given an encoding $\rho$}
is
$$ \minent(X;\rho) = -\log \sup_{F} \E_{x \sam X}[\Tr(F_{x}\rho(x))], $$
where the supremum ranges over all POVMs $F = \set{F_x}_{x \in \Lambda}$.
\end{definition}

We remark that there exists another definition of conditional min-entropy
in the quantum setting, which is more algebraic in flavor. However,
the two definitions are equivalent, as shown in~\cite{KRS09}.

\begin{proposition}[{\cite[Proposition 2]{KT08}}]
\label{prop:guessing-storage} If $\rho$ is a $b$-storage encoding of
$X$ then $\minent(X;\rho) \ge \minent(X)-b$.
\end{proposition}

We shall need the following standard lemmas regarding min-entropy
that can be found, e.g., in~\cite{R05}. The first lemma says that
cutting $\ell$ bits from a source cannot reduce the min-entropy by
more than $\ell$.
\begin{lemma}\label{lem:cutting-source}
Let $X=X_1 \circ X_2$ be a distribution over bit strings and $\rho$
be an encoding such that $\minent(X;\rho) \ge k$, and suppose that
$X_2$ is of length $\ell$. Let $\rho'$ be the encoding of $X_1$
defined by $\rho'(x_1) = \E_{x \sam (X|X_1=x_1)} [\rho(x)]$. Then,
$\minent(X_1;\rho') \ge k-l$.
\end{lemma}
\begin{proof}
Given any predictor $P'$ which predicts $X_1$ from $\rho'$, we can construct a predictor $P$ for $X$ (from $\rho$) as follows: $P$ simply runs $P'$ to obtain a prediction for the prefix $x_1$, and then appends it with a randomly chosen string from $\B^\ell$. Then,
\begin{eqnarray*}
 \Pr_{x_1\circ\, x_2 \sam X}[P(\rho(x_1 \circ x_2)) = x_1 \circ x_2] &=& \Pr_{x_1\circ\, x_2 \sam X}[P'(\rho(x_1 \circ x_2)) = x_1] \cdot 2^{-\ell} \\
 &=& \Pr_{x_1 \sam X_1}[P'(\rho'(x_1)) = x_1] \cdot 2^{-\ell}.
\end{eqnarray*}

Thus, if $\minent(X_1;\rho') < k-l$ then there would have been a
predictor which predicts $X$ with probability greater than $2^{-k}$
and this cannot be the case since $\minent(X;\rho) \ge k$.
\end{proof}

The second lemma says that if a source has high min-entropy, then
revealing a short prefix (with high probability) does not change
much the min-entropy. The lemma is a generalization of a well known
classical lemma.
\begin{lemma}
\label{lem:prefix} Let $X=X_1 \circ X_2$ be a distribution and
$\rho$ be an encoding such that $\minent(X;\rho) \ge k$, and suppose
that $X_1$ is of length $\ell$. For a prefix $x_1$, let $\rho_{x_1}$
be the encoding of $X_2$ defined by $\rho_{x_1}(x_2) = \rho(x_1
\circ x_2)$. Call a prefix $x_1$ \emph{bad} if $\minent(X_2 ~|~ X_1
=x_1 ; \rho_{x_1}) \le r$ and denote by $B$ the set of bad prefixes.
Then, $$\Pr[X_1 \in B] \le 2^{\ell} \cdot 2^{r} \cdot 2^{-k}.$$
\end{lemma}
\begin{proof}
Let the prefix $x_1' \in B$ be the one with the largest probability
mass. Then, $\Pr[X_1 = x_1'] \ge \Pr[X_1 \in B] \cdot 2^{-\ell}$.
For any $z\in B$, let $A_z$ denote the optimal predictor that
predicts $X_2$ from $\rho_{z}$, conditioned on $X_1 = z$. By the
definition of min-entropy, for any $z\in B$,
$$ \E_{x_2 \sam (X_2|X_1=z)} \Pr[A_z(\rho_z(x_2)) = x_2] \ge 2^{-r}.$$
In particular this holds for $z = x_1'$.

Now, define a predictor $P$ for $X$ from $\rho$ by
$$P(\rho(x)) = x_1' \circ A_{x_1'}(\rho(x)),$$
that is, $P$ simply ``guesses" that the prefix is $x_1'$ and then applies the optimal predictor $A_{x_1'}$.
The average success probability of $P$ is
\begin{eqnarray*}
\E_{x \sam X} \big[ \Pr[P(\rho(x)) = x]\big] &=& \E_{x_1 \sam X_1} \left[ \E_{x_2 \sam (X_2|X_1=x_1)} \left[
\delta_{x_1,x_1'} \cdot \Pr[A_{x_1'}(\rho_{x_1'}(x_2)) = x_2] \right]\right]\\
&=& \Pr[X_1 = x_1'] \cdot \E_{x_2 \sam (X_2|X_1=x_1')} \left[\Pr[A_{x_1'}(\rho_{x_1'}(x_2)) = x_2 ] \right] \\
&\ge& \Pr[X_1 \in B] \cdot 2^{-\ell} \cdot  2^{-r}
\end{eqnarray*}
On the other hand, since $\minent(X;\rho) \ge k$, the average
success probability of $P$ is at most $2^{-k}$. Altogether, $\Pr[X_1
\in B] \le 2^{\ell} \cdot 2^{r} \cdot 2^{-k}$.
\end{proof}

\subsection{Quantum-proof extractors}
\label{sec:def:ext}

We now define the three different classes of extractors against quantum adversaries that we deal with in this paper.
We begin with the most general (and natural) definition:

\begin{definition}
\label{def:ext-q-knowledge} A function $E:\B^n \times \B^t \to \B^m$
is a \emph{quantum-proof $(n,k,\eps)$ strong extractor} if for every
distribution $X$ over $\B^n$ and every encoding $\rho$ such that
$\minent(X;\rho) \ge k$,
$$
\tnorm{U_t \circ E(X,U_t) \circ \rho(X) -U_{t+m} \times \brho_X} \le \eps.
$$
\end{definition}

We use $\circ$ to denote correlated values. Thus, $U_t \circ E(X,U_t) \circ \rho(X)$ denotes the mixed state obtained by sampling $x \sam X, y \sam U_t$ and outputting $\ketbra{y,E(x,y)}{y,E(x,y)} \tensor \rho(x)$. Notice that all 3 registers are correlated. When a register is independent of the others we use $\times$ instead of $\circ$. Thus, $U_{t+m} \times \brho_X$ denotes the mixed state obtained by sampling $x \sam X, w \sam U_{t+m}$ and outputting $\ketbra{w}{w} \tensor \rho(x)$.

Next we define quantum-proof extractors for \emph{flat distributions}:

\begin{definition}
\label{def:ext-flat} A function $E:\B^n \times \B^t \to \B^m$ is a
\emph{quantum-proof $(n,f,k,\eps)$ strong extractor for flat
distributions} if for every \emph{flat} distribution $X$ over $\B^n$
with exactly $f$ min-entropy and every encoding $\rho$ of $X$ with
$\minent(X;\rho) \ge k$,
$$
\tnorm{U_t \circ E(X,U_t) \circ \rho(X) -U_{t+m} \times \brho_X} \le \eps.
$$
\end{definition}

We remark that in the classical setting every extractor for flat distributions is also an extractor for general distributions, since every distribution with min-entropy $k$ can be expressed as a convex combination of flat distributions over $2^k$ elements.

Finally we define extractors against quantum storage:

\begin{definition}
\label{def:ext-q-storage}
A function $E:\B^n \times \B^t \to \B^m$ is an \emph{$(n,k,b,\eps)$ strong extractor against quantum storage}
if for every distribution $X$ over $\B^n$ with $\minent(X) \ge k$ and every $b$-storage encoding $\rho$ of $X$,
$$
\tnorm{U_t \circ E(X,U_t) \circ \rho(X) -U_{t+m} \times \brho_X} \le \eps.
$$
\end{definition}

The next lemma shows it sufficient to consider only flat distributions when arguing about the correctness
of extractors against quantum storage.
\begin{lemma}
\label{lem:storage-flat}
If $E$ is \emph{not} an $(n,k,b,\eps)$ strong extractor against quantum storage
then there exists a set $X$ of cardinality $2^k$ and a $b$-storage encoding $\rho$ such that
$E$ fails on $(X;\rho)$, that is,
$$
\tnorm{U_t \circ E(X,U_t) \circ \rho(X) -U_{t+m} \times \brho_X} > \eps.
$$
\end{lemma}

\begin{proof}
We prove the contrapositive, i.e., we assume that $E$ works for flat distributions of min-entropy exactly $k$ and prove that
it also works for general distributions with at least $k$ min-entropy.

Suppose $X$ is a distribution with $\minent(X) \ge k$. Then $X$ can expressed as a convex combination of flat distributions $X_i$ each with $\minent(X_i) = k$. If $\rho$ is a $b$-storage encoding of $X$  then it is also a $b$-storage encoding of each of these flat distributions $X_i$. Thus, by assumption,
 $$\tnorm{U_t \circ E(X_i,U_t) \circ \rho(X_i) -U_{t+m} \times \brho_{X_i}} \le \eps.$$
 Now by convexity,
 $$\tnorm{U_t \circ E(X,U_t) \circ \rho(X) -U_{t+m} \times \brho_X} \le \eps,$$
 as desired.
\end{proof}

Combining this with Proposition \ref{prop:guessing-storage} we get:
\begin{lemma}
\label{lem:flat-to-storage}
Every quantum-proof $(n,f,k,\eps)$ strong extractor for flat distributions, is an $(n,f,f-k,\eps)$ strong extractor against quantum storage.
\end{lemma}

\subsection{Lossless condensers}

\begin{definition}[strong condenser]
\label{def:condenser}
A mapping $C:\B^n \times \B^d \rightarrow \B^{n'}$ is
an \emph{$(n,k_1) \to_\epsilon (n',k_2)$
strong condenser} if for every distribution $X$ with $k_1$ min-entropy,
$U_d \circ C(X, U_d)$ is $\epsilon$-close to a distribution with $d + k_2$ min-entropy.
\end{definition}

One typically wants to maximize $k_2$ and bring it close to $k_1$ while minimizing $n'$ (it can be as small as $k_1+O(\logeps)$) and $d$ (it can be as small as $\log ((n-k)/(n'-k)) + \logeps + O(1)$). For a discussion of the parameters, see \cite[Appendix B]{CRVW02}. We call the condenser \emph{lossless} if $k_2 = k_1$.

The property of lossless condensers that we shall use is the
following.
\begin{fact}[{\cite[Lemma 2.2.1]{TUZ01}}]
\label{fact:unique-neigh}
Let $C:\B^n \times \B^d \rightarrow \B^{n'}$ be an $(n,k) \to_\epsilon (n',k)$ lossless condenser. Consider the mapping $$C' : \B^n \times \B^d \rightarrow \B^{n'} \times \B^d$$ 
$$C'(x,y) = C(x,y) \circ y.$$
Then, for every set $X \subseteq \B^n$ of size $|X| \le 2^k$, there exists a mapping $C'': \B^n \times \B^d \rightarrow \B^{n'} \times \B^d$ that is injective on $X \times \B^d$ and agrees with $C'$ on at least $1-\eps$ fraction of the set $X \times \B^d$.
\end{fact}




\section{A reduction to full classical entropy}
\label{sec:reduction}

A popular approach for constructing explicit extractors in the classical setting is as follows:
\begin{itemize}
\item
Construct an explicit extractor for the \emph{high} min-entropy
regime, i.e. for sources $X$ distributed over $\B^n$ that have $k$
min-entropy for some large $k$ close to $n$, and,

\item
Show a reduction from the general case to the high min-entropy case.
\end{itemize}

In the classical setting this is often achieved by composing an
extractor for the high min-entropy regime with a classical lossless
condenser. Specifically, assume:
\begin{itemize}
\item
$C: \B^n \times \B^d \to \B^{n'}$ is an $(n,k) \to_{\eps_1} (n',k)$ strong lossless condenser, and,
\item
$E: \B^{d+n'} \times \B^t \to \B^m$ is a $(d+n',d+k,\eps_2)$ strong extractor.
\end{itemize}
Define $EC: \B^n \times (\B^d \times \B^t) \to \B^m$ by
\begin{eqnarray*}
EC(x,(y_1,y_2)) & = & E((C(x,y_1),y_1),y_2).
\end{eqnarray*}

In the classical setting, \cite[Section 5]{TUZ07} prove that $EC$ is
a strong $(n,k,\eps_1+\eps_2)$ extractor. In this section we try to
generalize this result to the quantum setting. We prove:

\begin{theorem}
\label{thm:reduction}
Let $C$ and $EC$ be as above.
\begin{itemize}
\item
If $E$ is a quantum-proof $(d+n',d+k,k_2,\eps_2)$ strong extractor for flat distributions, then
$EC$ is a $(n,k,k_2,\eps=\eps_2+2\eps_1)$ strong
extractor for flat distributions.
\item
If $E$ is a $(d+n',d+k,d+b,\eps_2)$ strong extractor against quantum storage,
then $EC$ is an $(n,k,b,\eps=\eps_2+2\eps_1)$ strong
extractor against quantum storage.
\end{itemize}
\end{theorem}

The intuition behind the theorem is the
following. When the condenser $C$ is applied on a flat source, it
is essentially a one-to-one mapping between the source $X$ and its
image $C(X)$. Therefore, roughly speaking, any quantum information
about $x$ can be translated to  quantum information about $C(x)$
and vice-versa. To make this precise we need to take care of the
condenser's seed, and this incurs a small loss in the parameters.

We first prove the second item.
\begin{proof}[ (second item)]
Assume, by contradiction that $EC$ is not an $(n,k,b,\eps=\eps_2+2\eps_1)$ strong
extractor against quantum storage. Then, by Lemma \ref{lem:storage-flat},
there exists a subset $X \subseteq \B^n$ of cardinality $2^{k}$ and a $b$-storage encoding $\rho$ of $X$ such that, given this encoding, the output of the extractor $EC$ is not $\eps$-close to uniform. That is,
\begin{eqnarray*}
\tnorm{U_{t+d} \circ EC(X,U_{t+d}) \circ \rho(X) -U_{t+d+m} \times \brho_X} & > & \eps.
\end{eqnarray*}

In particular, by Fact~\ref{fact:trace-norm-distinguisher}, there exists some Boolean measurement that $\eps$-distinguishes the two distributions. Since the first two components are classical, we can represent this measurement as follows. For every $y \in \B^{t+d}$ and $z \in \B^m$ there exists a Boolean measurement $\set{F^{y,z},I-F^{y,z}}$ on the quantum component such that
\begin{eqnarray*}
\biggl|\E_{x\sam X,~y \sam U} \bigl[\Tr \bigl(F^{y,EC(x,y)} \rho(x)\bigr)\bigr]  - \E_{y,z \sam U} \bigl[\Tr \bigl(F^{y,z} \brho_X\bigr)\bigr] \biggr| & > & \eps.\\
\end{eqnarray*}

We now show how this can be used to break the extractor $E$. Consider the set $A=X \times \B^d$. By Fact~\ref{fact:unique-neigh}, there exists a mapping $D$ that is injective on $A$ and agrees with the condenser on at least $1-\eps_1$ fraction of $A$.
Denoting $B = D(A)$, it is clear that $\minent(B) \ge d+k$.

For $(\tilde{x},\tilde{y}) \in B$ we define the encoding
$$\rho'(\tilde{x},\tilde{y}) = \ketbra{y_1}{y_1} \tensor \rho(D^{\leftarrow}(\tilde{x},\tilde{y})),$$
where $(x,y_1)=D^{-1}(\tilde{x},\tilde{y}) \in A$  is the unique element  such that $D(x,y_1) = (\tilde{x},\tilde{y})$, and
$D^{\leftarrow}(\tilde{x},\tilde{y})=x$.

Next, we define a measurement $\set{\ol{F}^{y_2,z},I-\ol{F}^{y_2,z}}$ that given the input $y_2 \in \B^t, z \in \B^m$ and $\rho'(\tilde{x},\tilde{y})= \ketbra{y_1}{y_1} \tensor \rho(x)$, sets $y=(y_1,y_2)$ and applies the measurement $\set{F^{y,z},I-F^{y,z}}$ on the quantum register $\rho(x)$.

Now,
\begin{eqnarray*}
\biggl|\E_{b \sam B,~y_2\sam U_t} \bigl[ \Tr \bigl( \ol{F}^{y_2,E(b,y_2)} \rho'(b)\bigr)\bigr] - \E_{x \sam X,~y \sam U_{d+t}} \bigl[ \Tr \bigl( F^{y,EC(x,y)} \rho(x)\bigr) \bigl] \biggr| & \le & \eps_1,
\end{eqnarray*}
since the flat distribution over $B$ is $\eps_1$-close to the distribution obtained by sampling $x \in X$, $y_1 \in U_d$ and outputting $(C(x,y_1),y_1)$. For the same reason, averaging over $B$ for $\ol{F}$ is almost as averaging over $X$ for $F$. Namely,
\begin{eqnarray*}
\biggl|\E_{y_2,z \sam U}  \bigl[ \Tr \bigl( \ol{F}^{y_2,z}  \bar{\rho'}_B \bigr) \bigr]  -\E_{y,z \sam U} \bigl[ \Tr \bigr( F^{y,z} \brho_X \bigr) \bigr] \biggr| & \le & \eps_1.
\end{eqnarray*}

It follows that
\begin{eqnarray*}
\biggl|\E_{b \sam B,~y_2 \sam U} \bigl[ \Tr \bigl( \ol{F}^{y_2,E(b,y_2)} \rho'(b)\bigr) \bigr]  - \E_{y_2,z \sam U}  \bigl[ \Tr \bigl( \ol{F}^{y_2,z}  \bar{\rho'}_B  \bigr) \bigr] \biggr| & \ge & \\
\biggl|\E_{x \sam X,~y\sam U} \bigl[ \Tr \bigl( F^{y,EC(x,y)} \rho(x) \bigr) \bigr]  -\E_{y,z \sam U} \bigl[ \Tr \bigl( F^{y,z} \brho_X  \bigr) \bigr] \biggr| -2\eps_1 & > & \eps-2\eps_1=\eps_2.
\end{eqnarray*}

Clearly $\rho'$ is a $(d+b)$-storage encoding of $B$.
This contradicts the fact that $E$ is a strong extractor against $d+b$ quantum storage.
\end{proof}

We now prove the first item.

\begin{proof}[ (first item)]
Assume, for contradiction, that $EC$ is not a quantum-proof
$(n,k,k_2,\eps)$ strong extractor for flat distributions. Then there
exists a subset $X \subseteq \B^n$ of cardinality exactly $2^{k}$
and an encoding $\rho$ of $X$ such that the conditional min-entropy
is at least $k_2$ but given this encoding the output of the
extractor $EC$ is not $\eps$-close to uniform. The proof proceeds as
before, defining the Boolean measurement $F$, the sets $A$ and $B$,
the encoding $\rho'$ and the measurement $\overline{F}$. If we can
show that $\minent(B;\rho') \ge k_2$ then we break the extractor $E$
and reach a contradiction.  Indeed:

\begin{claim}
$\minent(B; \rho') \ge k_2$.
\end{claim}
\begin{proof}
Assume, for contradiction, that $\minent(B;\rho') < k_2$. Then,
there exists a predictor $W'$ such that $$\Pr_{b \sam B} [W'
(\rho'(b))=b] >2^{-k_2}.$$ Define a new predictor, $W$, that given
$\rho(x)$ works as follows. First $W$ chooses $y \sam U_d$ and
 runs $W'$ on $\ketbra{y}{y} \tensor \rho(x)$ to get some
answer $\wt{b}$. It then outputs $D^{\leftarrow}(\wt{b})$.

The success probability of the predictor $W$ is
\begin{eqnarray*}
\Pr_{x \sam X} [W(\rho(x))=x] & = & \Pr_{x \sam X, y \in \B^d} [D^{\leftarrow}(W'(\ketbra{y}{y} \tensor \rho(x)))=x ] \\
& \ge &  \Pr_{x \sam X, y \in \B^d} [W'(\ketbra{y}{y} \tensor \rho(x)) = D(x,y) ] \\
& = & \Pr_{b \sam B} [W'(\rho'(b)) = b ] > 2^{-k_2}.
\end{eqnarray*}
This contradicts the fact that $\minent(X; \rho) \ge k_2$.
\end{proof}
\end{proof}

We remark that we do not know how to extend the proof to work with
lossy condensers.

%

\section{An explicit quantum-proof extractor for the high-entropy regime}
\label{sec:ext-less-than-half} In this section we describe a
construction of a short-seed quantum-proof $(n,k,\eps)$ strong extractor that works whenever
$k \gg n/2$. In the classical setting this scenario was studied
in~\cite{CRVW02}, developing and improving techniques
from~\cite{NZ96} and other papers. Here we only need the
techniques developed in~\cite{NZ96}.

Intuitively, the extractor $E$ that we construct works as follows.
First, it divides the source to two parts of equal length. Since the
min-entropy is larger than $n/2$, for almost any fixing of the first
part of the source, the distribution on the second part has
$\Omega(n)$ min-entropy. Hence, applying an extractor $E_2$ on the
second part results in output bits that are close to uniform. Since
this is true for almost every fixing of the first part, these output
bits are essentially independent of the first part of the source.
Therefore, these output bits can serve as a seed for another
extractor, $E_1$, that is applied on the first part of the source.

Formally, assume:
\begin{itemize}
\item
$E_1: \B^{n/2} \times \B^{d_1} \to \B^{m_1}$ is a quantum-proof $(\frac{n}{2},\frac{n}{2}-b,\eps_1)$ strong extractor, and,
\item
$E_2: \B^{n/2} \times \B^{d_2} \to \B^{d_1}$ is a quantum-proof $(\frac{n}{2},k,\eps_2)$ strong extractor.
\end{itemize}
Define $E: \B^n \times \B^{d_2} \to \B^{m_1}$ by
$$ E(x,y) = E_1(x_1, E_2(x_2,y)),$$
where $x =x_1 \circ x_2$ and $x_1,x_2 \in \B^{n/2}$.

\begin{theorem}
\label{thm:extracotr-composition} Let $E_1,E_2$ and $E$ be as above
with $k=\frac{n}{2}-b-\logeps$. Then $E$ is a quantum-proof
$(n,n-b,\eps+\eps_1+\eps_2)$ strong extractor.
\end{theorem}

%

\begin{proof}
Let $X=X_1 \circ X_2$ be a distribution on $\B^n = \B^{n/2} \times
\B^{n/2}$ and $\rho$ be an encoding such that $\minent(X;\rho) \ge
n-b$. For a prefix $x_1 \in \B^{n/2}$, let $\rho_{x_1}$ be the
encoding of $X_2$ defined by $\rho_{x_1}(x_2) = \rho(x_1 \circ
x_2)$. A prefix $x_1$ is said to be \emph{bad} if $\minent(X_2 ~|~
X_1 =x_1 ; \rho_{x_1}) \le k$. By Lemma~\ref{lem:prefix}, the
probability $x_1$ (sampled from $X_1$) is bad is at most
$$\frac{2^{n/2} \cdot 2^{k} }{2^{n-b}} = \frac{2^{n/2} \cdot 2^{n/2-b-\logeps} }{2^{n-b}} = \eps .$$

Whenever $x_1$ is not bad, $\minent(X_2 ~|~ X_1 =x_1 ; \rho_{x_1}) >
k$, that is, the extractor $E_2$ is applied on a distribution with
$k$ min-entropy. Therefore, by the assumption on $E_2$, its output
is $\eps_2$-close to uniform. That is, for every good $x_1$,
$$\tnorm{U_{d_2} \circ x_1 \circ E_2(X_2,U_{d_2}) \circ \rho_{x_1}(X_2) - U_{d_2} \circ x_1 \circ U_{d_1} \circ \rho_{x_1}(X_2)} \le \eps_2 .$$

Hence, the distribution $U_{d_2} \circ X_1 \circ E_2(X_2,U_{d_2}) \circ \rho(X)$ is $(\eps +\eps_2)$-close to $U_{d_2} \circ X_1 \circ U_{d_1} \circ \rho(X)$. In particular,
\begin{eqnarray*}
& & \tnorm{U_{d_2} \circ E(X,U_{d_2}) \circ \rho(X) -U_{d_2+d_1} \circ \brho_X}  \\
& = & \tnorm{U_{d_2} \circ E_1(X_1, E_2(X_2,U_{d_2})) \circ \rho(X) -U_{d_2+d_1} \circ \brho_X}  \\
& \le & \eps + \eps_2 + \tnorm{U_{d_2} \circ E_1(X_1, U_{d_1}) \circ \rho(X) -U_{d_2+d_1} \circ \brho_X},
\end{eqnarray*}
where the last inequality follows from Fact~\ref{fact:trace-norm-reduced-by-processing}.

Since, $\minent(X;\rho) \ge n-b$, by Lemma~\ref{lem:cutting-source},
if we define an encoding $\rho'$ of $X_1$ by $\rho'(x_1) = \E_{x\sam
(X|X_1=x_1)} [\rho(x)]$, then $\minent(X_1;\rho') \ge n-b-n/2 =
n/2-b$. Therefore, by the assumption on $E_1$ we get
$$\tnorm{E_1(X_1,U_{d_1}) \circ \rho(X)-U_{m_1} \tensor \brho_X} \le \eps_1,$$ and thus
\begin{eqnarray*}
\tnorm{U_{d_2} \circ E(X,U_{d_2}) \circ \rho(X) -U_{d_2+d_1} \tensor \brho_X}  & \le & \eps + \eps_1 + \eps_2 .
\end{eqnarray*}
\end{proof}

\subsection{Plugging in explicit constructions}

We use Trevisan's extractor, which was already shown to be quantum-proof in~\cite{DV10,DPVR09}. Specifically, we use the following two instantiations of this extractor:

\begin{theorem}[\cite{DPVR09}]\label{thm:Trevisan-polylog}
For every constant $\delta>0$, there exists $E_1: \B^{\frac{n}{2}}
\times \B^{O(\log^2(n/\eps_1))} \to
\B^{(1-\delta)(\frac{n}{2}-b)}$ which is a quantum-proof
$(\frac{n}{2},\frac{n}{2}-b,\eps_1)$  strong extractor.
\end{theorem}

\begin{theorem}[\cite{DPVR09}]\label{thm:Trevisan-log}
For every constants $\gamma_1, \gamma_2>0$, there exists $E_2:
\B^{\frac{n}{2}} \times \B^{O(\log(n/\eps_2))} \to
\B^{k^{1-\gamma_1}}$ which is a quantum-proof $(\frac{n}{2},k,\eps_2)$  strong
extractor, for $k>n^{\gamma_2}$.
\end{theorem}

Plugging these two constructions into
Theorem~\ref{thm:extracotr-composition} gives
Theorem~\ref{thm:main-qproof} which we now restate.
\newtheorem*{thmknow}{Theorem~\ref{thm:main-qproof}}
\begin{thmknow}
For any $\beta < \half, \gamma > 0$ and $\eps \ge
2^{-n^{(1-\gamma)/2}}$, there exists an explicit quantum-proof
$(n,(1-\beta)n,\eps)$ strong extractor
$E:\B^n \times \B^t \to \B^m$, with seed length $t=O(\log
n+\logeps)$ and output length $m=\Omega(n)$.
\end{thmknow}
\begin{proof}
We set $\eps_1=\eps_2=\eps$,~ $b=\beta n$,~ $k=\frac{n}{2}-\beta
n-\logeps$,~ $\gamma_2=\delta=\half$ and $\gamma_1 < \gamma$. In
order to apply Theorem~\ref{thm:extracotr-composition} we need to
verify that the output length of $E_2$ is not shorter than the
seed length of $E_1$. This is indeed the case since
$$k^{1-\gamma_1} \ge (\frac{n}{2}-\beta n-n^{\frac{
1-\gamma}{2}})^{1-\gamma_1} \ge n^{{1-\gamma}} \ge
O(\log^2(\frac{n}{\eps})).$$ The output length of $E$ is $\half
(\half - \beta)n = \Omega(n)$.
\end{proof}

\section{The final extractor for the bounded storage model}
\label{sec:final}
We need the classical lossless condenser of \cite{GUV07}.
\begin{theorem}[\cite{GUV07}]\label{thm:GUV} For every $\alpha>0$ there exists an $(n, k) \to_\epsilon ((1+\alpha)k, k)$ strong  lossless condenser $C$ with seed length $O(\log n+\logeps)$.
\end{theorem}

Plugging the condenser $C$ and the extractor $E$ of
Theorem~\ref{thm:main-qproof} into Theorem~\ref{thm:reduction}
gives Theorem~\ref{thm:main-storage}, which we now restate.
\newtheorem*{thmmainstorage}{Theorem~\ref{thm:main-storage}}
\begin{thmmainstorage}
For any $\beta < \half$ and $\eps \ge 2^{-k^\beta}$, there exists
an explicit $(n,k,\beta k,\eps)$ strong extractor against quantum
storage, $E:\B^n \times \B^t \to \B^m$,
with seed length $t=O(\log n+\logeps)$ and output length
$m=\Omega(k)$.
\end{thmmainstorage}
\begin{proof}
Let $\zeta > 0$ be a constant to be fixed later. The extractor $E$
from Theorem~\ref{thm:main-qproof}, when the source length is
set to be $2(1-\beta)(1-\zeta)k$, is a quantum-proof
$\big(2(1-\beta)(1-\zeta)k,(1-\beta)k,\eps \big)$ strong extractor. In particular, it is a
$\big(2(1-\beta)(1-\zeta)k,k,\beta k,\eps\big)$ strong extractor
against quantum storage. Its output
length is $\Omega(k)$. The theorem follows by applying
Theorem~\ref{thm:reduction}, using the condenser of
Theorem~\ref{thm:GUV} with $\alpha = 2(1-\beta)(1-\zeta)-1.$ Since
$\beta < \half$ there is a way to fix $\zeta$ such that $\alpha >
0$.
\end{proof}

Since Theorem~\ref{thm:reduction} works in the more general model of flat distributions, and since the extractor from Theorem~\ref{thm:main-qproof}
already works in the most general setting, we get Theorem \ref{thm:main-flat}:

\newtheorem*{thmmainflat}{Theorem~\ref{thm:main-flat}}
\begin{thmmainflat}
For any $\beta < \half$ and $\eps \ge 2^{-k^\beta}$, there exists an explicit quantum-proof $(n,k,(1-\beta) k,\eps)$ strong extractor for flat distributions, $E:\B^n \times \B^t \to \B^m$, with seed length $t=O(\log n+\logeps)$ and output length $m=\Omega(k)$.
\end{thmmainflat}

\noindent {\bf Acknowledgements.} We thank Roy Kasher for pointing
out an error in an earlier version of the paper. We thank
Christopher Portmann for helpful comments. We thank the anonymous
referees for many helpful suggestions that helped improve the paper.

\bibliographystyle{plain}
\bibliography{refs}
\end{document}